\documentclass[letterpaper,11pt,prx,aps,superscriptaddress,showpacs,nofootinbib]{revtex4-2}

\usepackage[utf8]{inputenc}
\usepackage[colorlinks=true,linkcolor=blue,citecolor=blue,urlcolor=blue]{hyperref}
\usepackage{graphicx}
\usepackage{dsfont}
\usepackage{xcolor}
\usepackage{amsthm}
\usepackage{float}
\usepackage{physics}
\usepackage{mathtools}
\usepackage{subcaption}
\usepackage{amsmath,amssymb}



\newtheorem{lemma}{Lemma}
\newtheorem{definition}{Definition}
\newtheorem{proposition}{Proposition}
\newtheorem{theorem}{Theorem}

\theoremstyle{remark}

\definecolor{FGreen}{RGB}{1,68,33}

\newcommand{\idty}{\mathds{1}}

\newcommand{\cket}[1]{\left | #1 \right )}
\newcommand{\cbra}[1]{\left ( #1 \right |}
\newcommand{\cbraket}[2]{\left ( #1 | #2 \right ) }

\newcommand{\cdyad}[1]{\left \vert #1 \middle) \middle( #1 \right \vert }

\newcommand{\cdyadm}[2]{\left \vert #1 \middle) \middle( #2 \right \vert }

\newcommand{\C}{\mathbb{C}}
\newcommand{\R}{\mathbb{R}}
\newcommand{\Z}{\mathbb{Z}}

\newcommand{\cE}{\mathcal{E}}
\newcommand{\cL}{\mathcal{L}}
\newcommand{\cP}{\mathcal{P}}
\newcommand{\cS}{\mathcal{S}}
\newcommand{\cV}{\mathcal{V}}
\newcommand{\cO}{\mathcal{O}}

\newcommand{\cT}{\mathcal{T}}

\newcommand{\singlesites}{\cP_{\text{single-site}}}

\newcommand{\equivclass}{{\tilde{O}}}
\newcommand{\ppath}{{\gamma}}
\newcommand{\ppathset}{{\Gamma}}

\newcommand{\poly}{\text{poly}}

\newcommand{\syndromepb}[1]{p(#1)}

\newcommand{\synd}[1]{\sigma_{#1}}
\newcommand{\bnb}[1]{\mathcal{B}(#1)}
\newcommand{\hbt}{\mathcal{H}}

\newcommand{\syndvec}{\vec{s}}

\makeatletter
\def\l@subsection#1#2{}
\def\l@subsubsection#1#2{}
\makeatother

\begin{document}

\title{The Free Energy Barrier: An Eyring–Polanyi bound for stabilizer Hamiltonians, with applications to quantum error correction}

\author{Nou\'edyn Baspin}
\affiliation{Iceberg Quantum, Sydney}
\affiliation{Centre for Engineered Quantum Systems, School of Physics, University of Sydney, Sydney, NSW 2006, Australia}

\begin{abstract}
The lack of an energy barrier in stabilizer Hamiltonians is well known to be an indication of short thermalisation times; and serves as a simple criteria to rule out self-correction. Its applicability has recently been challenged by the discovery of stabilizer codes with maximal energy barriers whose self-correcting status remains unclear. We address this question by introducing the more general notion of a \emph{free energy barrier}, whose absence is also demonstrated to guarantee fast thermalisation. Applying these new results to the Layer codes, we demonstrate that they lack self-correction.
\end{abstract}

\maketitle

\section{Introduction}

When an open quantum system interacts with an environment at fixed temperature, it is generally understood that, over long periods of time, the system invariably converges towards the corresponding Gibbs state. The speed of that convergence, however, can vary drastically from system to system; and understanding the factors at play is a wide open question. 

In this work we study the thermalisation time of some systems subject to Davies generators \cite{Davies1974, Davies1976, roga2010davies}. These represent the most common model of system-bath interactions in quantum systems, and have been the gold standard to study associated thermalisation processes \cite{rivas2012quantum, alicki2008thermal, alicki2009thermalization}. We will restrict our attention to stabilizer Hamiltonians, as they are singularly relevant to quantum error correction. 

Temme \cite{temme2016thermalization} showed that systems whose energy barrier is less than $\bar{\epsilon}$ thermalise in time at most $\cO(\beta n^2 e^{\bar{\epsilon}})$ at inverse temperature $\beta$; and this result provided theoretical guarantees on the applicability of Arrhenius' law as an upper bound on the mixing time of such a system. This result constitutes a vital theoretical instrument in the study of self-correcting quantum memories. In such a scheme, information is stored in the ground space of a Hamiltonian, whose dynamics are engineered to shield the information from thermal fluctuations. Naturally, the amount of time during which the encoded information is recoverable is upper bounded by the time it takes for the system to completely thermalise\footnote{These times can differ pretty significantly. For example, the 3D toric code is usually understood to have a small coherence time because its string-type logicals get quickly corrupted; while it also exhibits a long thermalisation time as its brane-type logicals are well protected from thermal fluctuations.}. 

Beyond the Arrhenius-Temme bound, our understanding of quantum memories is very limited. Energy barriers are necessary, but could they also be sufficient? Our interest in this question was sparked by the discovery of the Layer Codes \cite{williamson2024layer}: a family of stabiliser codes that exhibit the highest possible energy barrier for 3D local stabilizer Hamiltonians. Their energy barrier is sufficiently large that the Arrhenius-Temme bound does not \emph{a priori} preclude these codes to store information for super-polynomially long times. In this work, we bring this investigation to a conclusion, and show that these codes cannot be self correcting.

Our main results are articulated around the notion of {free energy barrier}\footnote{See Definition \ref{def:free-energy}} and generalise \cite{temme2016thermalization}. We show that the lack of a free energy barrier also bounds the thermalisation time of the system.

\begin{theorem}[Informal restatement of Theorem \ref{thm:final}]
\label{thm:intro-main}
    Let $H$ be a stabilizer Hamiltonian whose free energy barrier is less than $\bar{f}$ at temperature $\beta$. Then its thermalisation time $t_{\mathrm{mix}}$ satisfies:
    \[
    t_{\mathrm{mix}} \leq \cO(\beta n^2 e^{\beta \bar{f}})
    \]
\end{theorem}

The appeal of this new notion resides in situations where $\bar{f} \ll \bar{\epsilon}$ and Theorem \ref{thm:intro-main} is tighter than the Arrhenius-Temme bound. In Section \ref{sec:layer-codes} we show this to be the case for the Layer codes: while they obey $\bar{\epsilon} \propto n^{1/3}$ \cite{williamson2024layer}, we can bound $\bar{f}$ by a constant -- and thus their coherence time is at most polynomial.

\subsection{Open Questions}
\begin{enumerate}
    \item \label{item:upper-bound}The Bravyi-Terhal \cite{bravyi2009no} bound shows a tight $\bar{\epsilon} \leq \cO(n^{1/3})$ bound for the energy barrier of 3D local codes. Can their result be improved upon by showing an $\bar{f} \ll n^{1/3}$ bound on the free energy barrier of 3D codes \cite{siva2017topological}? Can our bound for the Layer codes be generalised to other 3D local codes with high energy barriers? \cite{portnoy2023localquantumcodessubdivided, lin2024geometricallylocalquantumclassical}

    \item What is the link between our definition of free energy barrier, and that of \cite{rakovszky2024bottlenecksquantumchannelsfinite, placke2024topologicalquantumspinglass}? Could it be leveraged to address Question \ref{item:upper-bound}?

    \item Is it possible to define a free energy barrier for a subalgebra of the $n$-qubit Pauli algebra, and can we show thermalisation in that subsystem alone? This could provide a way to capture decoherence without full thermalisation -- as in the case of the aforementioned 3D toric code. On a related topic, the traditional definition of an energy barrier \cite{bravyi2009no} only concerns \emph{logical} operators, instead of the entirety of $\cP^n$; could that definition be sufficient to show thermalisation within the ground space? For example, can we show that the subspace associated with $\Pi_{\syndvec}$ thermalises quickly if $\tau_{\tilde{O}}(\syndvec)$ is bounded?
    
    \item Is our bound tight? The framework we use to show the bound on the mixing time is an adaptation of the canonical paths method \cite{guruswami2016rapidly}. In its classical setting, it is known that a quantity similar to $\norm{\ket{w_{(U_eQ, V_eQ)}(\syndvec)}}_2^2$ in Proposition \ref{prop:row-norm} also leads to a \emph{lower} bound on the mixing time -- see Theorem 4.8 of \cite{guruswami2016rapidly} and Theorem 8 of \cite{sinclair1992improved}. We suspect \mbox{Theorem \ref{thm:final}} might be tight up to factors polynomial in $n$.

    \item The support number bound (see Theorem \ref{thm:support-nb-mixing}), most commonly referred to as the spectral gap bound \cite{temme2013lowerbounds}, can often be improved on by using the log-Sobolev constant \cite{kastoryano2013quantum}; which offers the potential of reducing the bound on the mixing time from $\poly(n) \beta e^{\beta \bar{f}}$ to $\log(n) \beta e^{\beta \bar{f}}$. In order to realise this improvement, one would likely need to adapt Theorem 2.14 of \cite{montenegro2006mathematical} to quantum Liouvillians. See also \cite{roberto2003path} for an attempt at adapting the canonical paths method to the log-Sobolev constant.

\end{enumerate}

\subsection{Related works}

We would like to highlight that a bound similar to Theorem \ref{thm:intro-main} was conjectured (though not proved) in Equation 65 of \cite{siva2017topological}, and served as a guiding light in the early days of this project. To formalise this conjecture, we rely on the framework established by \cite{temme2016thermalization}, which is an adaptation of the canonical path method -- see \cite{guruswami2016rapidly} for a gentle introduction, Chapter 5 of \cite{jerrum2003counting}, and Chapter 12.3 of \cite{jerrum1996markov}. The tools we use also owe much to \cite{alicki2009thermalization}, one of the earliest formal studies of the thermalisation time of quantum memories. 

We note that despite conceptually following \cite{temme2016thermalization}, we will often deviate from it when it comes to notation. This is done to align with conventions more commonly used in the QEC literature, with the hope of making the proofs accessible to the largest audience. 

\section{Free energy barrier of stabilizer Hamiltonians}

This work focuses on stabilizer Hamiltonians, which take their terms from the Pauli matrices ${\cP^n = \{\idty, X, Y, Z\}^{\otimes n}}$. We will occasionally also refer to the `single-site' Pauli matrices ${\cP_{single-site} \subset \cP^n}$; they consist of the Pauli matrices that act as the identity on all but one qubit. The group commutator between $A,B \in \cP^n$ is expressed as ${[A,B] = ABA^{-1}B^{-1}}$. The Hilbert space these operators act on is $\hbt = (\C^2)^{\otimes n}$, and we write $\bnb{\hbt}$ the set of bounded linear operators on $\hbt$. 

\begin{definition}[Stabilizer Hamiltonian]
\label{def:stabiliser-ham}
    We say that $H$ is a stabilizer Hamiltonian on $n$ qubits when
    \[
    H = - \sum_i J_i g_i
    \]
    where $J_i$ is a positive real number and $\{g_i\}_i$ is an abelian subset of $\cP^n$.

\end{definition}

The abelian group $\cS$ generated by the set of operators $\{g_i\}_i$ is called the \emph{stabilizer group} -- note that since $ - \idty \not\in \{g_i\}_i$ and it is abelian, then $ - \idty \not\in \cS $. The Hamiltonian is then said to encode $k = n-r$ logical qubits, where $r$ is the rank of the group $\cS$. Due to the abelian nature of $\cS$ we have the isomorphism $\cS \cong \Z_2^r$, or $\cS = \langle S_1, S_2, \dots, S_r \rangle$, for some set $\{S_1, \dots, S_r\}$ of Pauli operators \footnote{Note that while $\{S_1, \dots, S_r\}$ forms an independent basis for $\cS$, $\{g_i\}_i$ does not need to. }. These stabilizers allow us to assign a `syndrome' to every Pauli; the syndrome map $\sigma_P = ([P, S_1], [P,S_2], \dots, [P,S_r]) \in \{-1,1\}^{\otimes r}$ spits out a vector for every $P \in \cP^n$. For two syndrome vectors $\syndvec, \vec{u}$ we write $\syndvec \oplus \vec{u} \coloneq (s_1\cdot u_1, s_2 \cdot u_2, \dots, s_r \cdot u_r)$. This allows us to split the Hilbert space into $2^r$ syndrome subspaces, and the projectors onto these subspaces can be written as:
\begin{equation}
    \Pi_{\syndvec} = \prod_j \frac{1}{2}(\idty + s_j S_j ), \quad \syndvec \in \{-1,1\}^{\otimes r}
\end{equation}
We can then make use of this new basis to neatly diagonalize the Hamiltonian $H$. To each syndrome vector $\syndvec$ we can associate a number $\epsilon_{\syndvec}$, such that:
\begin{equation}
    H = \sum_{\syndvec \in \{-1,1\}^{\otimes r}}     \epsilon_{\syndvec} \, \Pi_{\syndvec}
\end{equation}
Writing $p(\syndvec) \coloneq \frac{e^{-\beta \epsilon_{\syndvec}}}{\tr(e^{-\beta H})}$, the Gibbs state then becomes:
\begin{equation}
    \rho_\beta = \frac{e^{-\beta H}}{\tr(e^{-\beta H})} = \sum_{\syndvec} p(\syndvec) \Pi_{\syndvec}
\end{equation}
The notion of energy barrier of \cite{bravyi2009no, temme2016thermalization, michnicki2012welded, williamson2024layer} attempts to quantify the energy cost associated with transitioning from the state $\ket{\psi}$ to the state $\ket{\psi'} \coloneq P \ket{\psi}$ for some Pauli $P \in \cP^n$. This transition is assumed to take place along a path defined by a sequence of Pauli operators $(\idty, U_1, U_2,  \dots, P)$; where $U_i^\dagger U_{i+1} \in \cP_{single-site}$. 
The steps of this transition correspond then to $\ket{\psi} \rightarrow U_1\ket{\psi} \rightarrow U_2 \ket{\psi} \rightarrow \dots \rightarrow P\ket{\psi} = \ket{\psi'}$. 
Depending on the nature of the path, the intermediary states might have more energy with respect to $H$ than the initial state: for this transition to occur, some energy cost has to be paid. Much like the least action principle, we aim to pick the path that incurs the smallest increase. Mathematically, a step $U$ is then said to cost $\bar{\epsilon}_{\ket{\psi} \rightarrow \ket{\psi'}} (U)$ units of energy, where:
\begin{equation}
    \bar{\epsilon}_{\ket{\psi} \rightarrow \ket{\psi'}} (U)= \tr(HU^\dagger \dyad{\psi} U) - \tr(H \dyad{\psi})
\end{equation}
However, if $\tr(H \dyad{\psi'}) \gg \tr(H \dyad{\psi})$ it is inevitable that $\bar{\epsilon}_{\ket{\psi} \rightarrow \ket{\psi'}}$ blows up, even for an `optimal' path; and conversely, $\bar{\epsilon}_{\ket{\psi'} \rightarrow \ket{\psi}}$ will eventually trend negative. We can solve our issue by having these positive and negative anomalies cancel each other out. A more proper definition of an energy barrier between states $\ket{\psi}$ and $\ket{\psi'}$ is then:
\begin{align*}
    \bar{\epsilon}_{\ket{\psi} \rightleftarrows \ket{\psi'}}(U) &= \bar{\epsilon}_{\ket{\psi} \rightarrow \ket{\psi'}}(U) + \bar{\epsilon}_{\ket{\psi'} \rightarrow \ket{\psi}}(U) \\
    &= \tr(\dyad{\psi}(P^\dagger U^\dagger H U  P  + U^\dagger H U  - P^\dagger H P - H)) \\
    & \leq \norm{P^\dagger U^\dagger H U  P   + U^\dagger H U  - P^\dagger H P - H}_\infty
\end{align*}
The last row allows us to have a state-independent bound on the energy barrier. This motivates the definition\footnote{This definition differs from that of \cite{temme2016thermalization} by a factor $2$.}:
\begin{equation}
    \bar{\epsilon}_P(U) \coloneq \norm{P^\dagger U^\dagger H U  P   + U^\dagger H U  - P^\dagger H P - H}_\infty
\end{equation}
The structure of stabilizer Hamiltonians allows us to derive an easily computable upper bound on $\bar{\epsilon}_P(U)$:
\begin{proposition} Let $H$ be a stabilizer Hamiltonian according to Definition \ref{def:stabiliser-ham}, then we have the following upper bound on $\bar{\epsilon}_P(U)$
    \[
    \bar{\epsilon}_P(U) \leq 4\sum_{i: [g_i, P] = 1, [g_i, U] =-1} |J_i| 
    \]
\end{proposition}
\begin{proof}
The structure of $H$ allows us to simplify the expression of $\bar{\epsilon}_P(U)$:
    \begin{align*}
    P^\dagger U^\dagger H U  P  + U^\dagger H U  - P^\dagger H P - H
    &=   - \sum_i J_i P^\dagger U^\dagger g_i U P - \sum_i J_i U^\dagger  g_i U  +  \sum_i J_i P^\dagger g_i P +  \sum_i J_i  g_i   \\
    &=   \sum_i J_i g_i \left(1 - [ g_i,U] + [g_i, P] - [g_i, UP] \right) \\
    &=   \sum_i J_i g_i \left(1 - [ g_i,U] + [g_i, P] - [g_i, U]\cdot [g_i, P] \right) \\
    &=   \sum_{i} J_i g_i (1 - [ g_i,U] ) (1+ [g_i, P]  )
    \\
    &=   2\sum_{i: [g_i, P] = 1} J_i g_i (1 - [ g_i,U] )
    \\
    &=   4\sum_{i: [g_i, P] = 1, [g_i, U] =-1} J_i g_i 
\end{align*}
The claimed bound results from the triangle inequality, and the fact that $\norm{g_i}_\infty = 1$
\end{proof}

Our improvement on the bound of \cite{temme2016thermalization} relies on the idea of using \emph{multiple} Pauli paths, while controlling the energy barrier associated with each of these paths. The necessity of considering a multiplicity of paths naturally leads to the definition of a \emph{Pauli flow}.

\begin{definition}[Pauli flow]
\label{def:pauli-flow}
     Consider a stabilizer Hamiltonian $H$ on $n$ qubits. A \emph{flow} for this Hamiltonian is a collection of paths $ \Gamma_{P} = \{ \ppath \}_\ppath$ for every Pauli $P \in \cP^n$, such that for every path $\ppath \in \Gamma_P$:

     \begin{enumerate}
         \item $\ppath = \left(U_{0}, U_{1}, U_{2}, \dots, U_{T} \right )$
         \item By convention $\gamma$ starts at the identity, $U_0 = \idty$
         \item The operators $U_{t-1}$ and $U_t$ only differ by a single-site Pauli, $U_t = \alpha \cdot U_{t-1}, \alpha \in \singlesites$ 
         \item The path $\gamma$ ends at the target Pauli, $U_T = P$
     \end{enumerate}
     An edge $e$ is a sequence $e=(U_e,V_e)$ where both $U_e, V_e \in \cP^n$, and $U_e^\dagger V_e \in \cP_{single-site}$; and we write $e \in \gamma = \left(U_{0}, U_{1}, U_{2}, \dots, U_{T} \right )$ if there exists $t_*$ such that $U_e = U_{t_*}, V_e = U_{t_*+1}$. In simple terms an edge is a pair of consecutive Paulis in a path.

     Finally, we write $|\gamma| = T$ the length of a path. Naturally, we can take $T \leq n$.
\end{definition}

In \cite{temme2016thermalization} it is shown that given a unique canonical path per Pauli $P$, i.e. $|\Gamma_P| = 1$, the thermalisation time is at most $t_{\text{mix}} \leq \cO(\beta n^2 e^{\bar{\epsilon}})$. One can reexpress this as a ``thermalisation rate" of $1/t_{\text{mix}} \geq 1/\cO(\beta n^2 e^{\bar{\epsilon}})$. If $|\Gamma_P|$ is allowed to grow, we might then approximate that the thermalisation event happens as many times as there are paths, i.e. the rate would transform multiplicatively: $ 1/t_{\text{mix}} \geq |\Gamma_P| \cdot \cO(\beta n^2 e^{\bar{\epsilon}})$

This musing hopefully justifies why $|\Gamma_P|$ of paths serves as natural proxy for multiplicity. 
Writing $\Omega = |\Gamma_P|$, it is reasonable to expect a bound of the form $t_{\mathrm{mix}} \lesssim poly(n) e^{\beta\bar{\epsilon}}/\Omega = poly(n) e^{\beta\bar{\epsilon} - \log(\Omega)}$, where $\bar{f} \coloneq \bar{\epsilon} - \beta^{-1} \log(\Omega)$ serves as a natural free energy analogue. 

However, as it stands this derivation is deeply flawed, and the resulting bound not well posed. In order to be meaningful $\Omega$ needs to discriminate between paths that are `similar', which calls for a more careful definition of the free energy associated with a given flow.

\begin{definition}[Free Energy Barrier]
\label{def:free-energy}
     A Pauli flow is said to have free energy at most $\bar{f}$ at inverse temperature $\beta$ when

    \[
        \max_P \max_{e : e\in \gamma, \gamma \in \Gamma_P} \ \bar{\epsilon}_P(U_e) + \frac{1}{\beta}\log\left(\frac{1}{\Omega_{P}(e) } \right)\leq \bar{f}
    \]

     where 
    \[
    \frac{1}{\Omega_{P}(e) } \coloneq \frac{|\Gamma_P^{\cap e}|}{  |\Gamma_{P} |  }, \quad \Gamma_P^{\cap e} = \{\gamma : \gamma \in \Gamma_{P}, \gamma \ni e\} 
    \]

    The Hamiltonian $H$ is said to have \emph{free energy barrier} $f$ when it is not possible to find a flow with free energy lower than $f$.
\end{definition}

\textbf{Remark:}
Although Definition \ref{def:pauli-flow} technically requires to specify a set of paths for every Pauli $P \in \cP^n$, Lemma \ref{lem:pauli-flow-degen} lets us leverage the degeneracy of $H$ and only specify a path for every equivalence class $\cP^n / \cS$, at the cost of a small free-energy penalty.

\section{Thermalisation model: Davies Generator}
\label{sec:model}
Thermalisation can be characterised as the long term process during which a system is connected to a bath of particles at inverse temperature $\beta$. We assume that evolution of the system can be described by a Markovian map $\Lambda_t$ (i.e. $\Lambda_{t_1 + t_2} = \Lambda_{t_2} \circ \Lambda_{t_1}$), such that the state of the system after time $t$ is given by $\rho_t = \Lambda_t(\rho)$ -- we leave the dependence of $\Lambda_t$ on $\beta$ implicit. As we require the map be both continuous over time and Markovian, it obeys \cite{manzano2020short, roga2010davies, alicki2009thermalization, rivas2012quantum}:
\begin{equation}
    \Lambda_t (\rho) = e^{t \cL (\rho)}, \quad   \lim_{t \rightarrow \infty} \Lambda_t(\rho) = \rho_\beta
\end{equation}

For some $\cL: \bnb{\hbt} \rightarrow \bnb{\hbt}$. The Gibbs state $\rho_\beta = e^{-\beta H}/\tr(e^{-\beta H})$ corresponds to the stationary state of the thermalisation process \cite{alhambra2023quantum}. The choice of $\cL$ corresponds to a specific choice of assumptions regarding how the system is connected to a thermal bath \cite{mozgunov2020completelypositive}. Here we elect to use the `Davies generator', as it has served as the standard thermalization model in previous investigations regarding self-correcting quantum memories \cite{alicki2008thermal}.

The exact form of the Liouvillian $\cL$ is a little bit simpler to describe in the Heisenberg picture. Remember that with respect to the Hilbert-Schmidt inner product, the adjoint of $\cL$ is the map $\cL^*$ such that $\tr(A^\dagger \cL(B)) = \tr(\cL^*(A)^\dagger B)$.
In practice, we mostly consider situations where $A$ and $B$ are Hermitian operators. Considering the Hermitian-preserving nature of $\cL$ (See \cite{kasatkin2023whichdifferential}, Proposition 4) we obtain $\tr(A \cL(B)) = \tr(\cL^*(A)B)$.
We can now express the Liouvillian as \cite{alicki2009thermalization}:
\begin{equation}
    \cL^*(f) = i[H,f] + \sum_{\omega, \alpha} h^\alpha(\omega) \left( {S_\omega^\alpha}^\dagger f S_\omega^\alpha - \frac{1}{2}\{{S_\omega^\alpha}^\dagger{S_\omega^\alpha}, f\} \right )
\end{equation}
\begin{equation}
    S_\omega^\alpha = \sum_{\syndvec: \omega^\alpha(\syndvec) = \omega} \alpha\  \Pi_{\syndvec}, \quad \omega^\alpha(\syndvec) = \epsilon_{\syndvec} - \epsilon_{\syndvec \oplus \sigma(\alpha)}
\end{equation}
The exact form of the functions $h^\alpha(\omega)$ depends on the microscopic details of the bath and is largely irrelevant to the rest of our discussion. We are only interested in the fact they satisfy the so-called KMS condition \footnote{Which ensures the Gibbs state is the unique fixed state of $\Lambda_t$.}: $h^\alpha(-\omega) = h^\alpha(\omega)e^{-\beta \omega}$.
It is generally taken that there exist constants $c_\circ,C_\circ$ independent of the system size such that $0 < c_\circ \leq h^\alpha(\omega) \leq C_\circ$ \cite{alicki2008thermal, alicki2009thermalization}.

We recall the notions of Dirichlet form and variance of $\cL$, as they will take a central role in the rest of this work.
\begin{definition}[Dirichlet form and variance]
\label{def:forms}
    For a Liouvillian $\cL$, we define $\cV$ the variance, and $ \cE$ the Dirichlet form of $\cL$:
    \[
    \cV(f) = \tr(\rho_\beta f^\dagger f) - |\tr(\rho_\beta f)|^2, \quad \cV: \bnb{\hbt}\rightarrow \C
    \]
    \[
    \cE(f) = - \tr(\rho_\beta f^\dagger \cL^*(f)), \quad  \cE: \bnb{\hbt}\rightarrow \C
    \]
\end{definition}

\subsection{Support number}

The method we will use to bound the thermalisation time of the Liouvillian $\cL$ is derived from analysing its spectrum. In the above section, we presented $\cL$ as mapping operators to operators, and it might not be clear what its `spectrum' might be referring to. Before progressing any further, we will describe how $\cL$ can be explicitly expressed as a linear operator over Pauli matrices (also referred to as a `superoperator'), whose spectrum is then a lot more intuitive to grasp.

Pauli operators are a basis for the space of complex matrices. Any operator $O : {\C^2}^n \rightarrow  {\C^2}^n$ can be represented as a vector $\cket{O} \in {\C^{2^n}}\otimes  {\C^{2^n}}$. Therefore any linear operator $\cT: O \rightarrow O'$ can represented as a matrix $\cT: {\C^{2^n}}\otimes  {\C^{2^n}} \rightarrow {\C^{2^n}}\otimes  {\C^{2^n}}$. The explicit form of these vectors can be expressed as:
\begin{equation}
    \cket{O} = \frac{1}{2^n} \sum_{P \in \cP^n} \tr(P O) \cket{P}
\end{equation}
\begin{equation}
    \hat{\cT} = \frac{1}{2^n} \sum_{P,P' \in \cP^n}\tr(P'\cT(P))\cdyadm{P'}{P} 
\end{equation}
The inner product of these vectors is related to the Hilbert-Schmidt inner product of the corresponding operators\footnote{By using the normalised Pauli operators $\{\idty/\sqrt{2}, X/\sqrt{2}, Y/\sqrt{2}, Z/\sqrt{2}\}^{\otimes n}$, one would instead obtain a basis such that $\tr(A^\dagger B) = \cbraket{A}{B}$. We decided to keep with the usual matrices for the sake of simplicity. See \cite{innocenti2023shadow} for more on operator frames.}:
\begin{align*}
    \tr(A^\dagger B) &= \frac{1}{2^n\cdot 2^n} \tr( \sum_{P,P' \in \cP^n} \tr(A^\dagger P)P \tr(BP') P') \\
    &=\frac{1}{2^n \cdot 2^n} \sum_{P,P' \in \cP^n} \tr(AP)^* \tr(BP') \tr(PP') \\
    &= \frac{2^n}{2^n \cdot 2^n} \sum_{P \in \cP^n} \tr(A P)^* \tr(BP) \\
    &= 2^n \cbraket{A}{B}
\end{align*}

We are now in position to define the support number of $\cL$. Intuitively, the support number is equivalent to $\lambda^{-1}$, with $\lambda$ the spectral gap of $\cL$. The primary conceptual difference between these object is that the support number is instead defined in function of the Dirichlet form and the variance.

\begin{definition}[Support number] 
\label{def:support-number}
For an ordered pair $(\hat{\cE} , \hat{\cV})$ of matrices, the support number is defined as:
    \[
    \min \{\tau \in \R: \tau\hat{\cE} - \hat{\cV} \geq 0\}
    \]
    For a Lindbladian $\cL$ whose Dirichlet form and variance (respectively) correspond to $\hat{\cE}$ and $\hat{\cV}$, then $\tau$ is said to be the support number of $\cL$.
\end{definition}

The support number can then be linked to the thermalisation.

\begin{theorem}[\cite{temme2013lowerbounds, temme2016thermalization}]
\label{thm:support-nb-mixing}
    Let $\cL$ be a Liouvilian with stationary state $\rho_\beta$ and support number $\tau$, then the following bound holds:
    \[
    \forall \sigma_0, \quad \norm{e^{t\cL}(\sigma_0) - \rho_\beta}_1 \leq \sqrt{\norm{\rho^{-1}_\beta}_\infty} e^{-t/\tau}
    \]
    Where $\sigma_0$ is a density matrix.
\end{theorem}
Observe that for qubit systems, the Gibbs state satisfies $\norm{\rho_\beta^{-1}}_\infty \leq e^{c_0\beta\norm{H}_\infty}$ for a universal constant $c_0$. With the help of some algebraic manipulations, Theorem \ref{thm:support-nb-mixing} guarantees that:
\begin{equation}
\label{eq:mixing-time}
    \forall t_* \geq t_{\mathrm{mix}} \coloneq \tau \left( c_0 \beta\norm{H}_\infty \allowbreak + \log(4) \right), \quad \norm{e^{t_*\cL}(\sigma_0) - \rho_\beta}_1 \leq \frac{1}{4}
\end{equation}

This mixing time $t_{\mathrm{mix}}$ can be understood as an upper bound on the coherent time of a system. Once passed that threshold it is impossible to distinguish between $e^{t_*\cL}(\sigma_0)$ and $e^{t_*\cL}(\sigma_1)$ with an accuracy greater than 50\%. 

\section{Thermalisation time bound}

As the previous section suggests, our primary objective is to upper bound the support number $\tau$, which itself will upper bound the thermalisation time $t_{\mathrm{mix}}$. However the minimisation process described by Definition \ref{def:support-number} is impractical. We will instead exploit a more amenable bound, first used by \cite{temme2016thermalization}.
\begin{lemma}[\cite{boman2003support,chen2005obtaining}]
\label{lem:factorisation-bound}
    Let $\hat{\cE} = A A^\dagger$, and $\hat{\cV}= BB^\dagger$ be a pair of matrices. Further, let $AW=B$. Then the support number $\tau$ of the pair $(\hat{\cE} , \hat{\cV})$ is bounded as:
    \[
    \tau \leq \max_m \sum_{k: W_{km}\neq 0} \norm{\ket{w_k}}^2_2
    \]
    Where we use the notation $W \coloneq \sum_{k=1}^K \sum_m W_{k,m}\dyad{k}{m}$ and $\ket{w_k}$ denotes the row vectors $\ket{w_k} = \sum_{m=1}^M W_{k,m} \ket{m}$.
\end{lemma}

Given a stabilizer subgroup $\cS \subset \cP^n$ on $n$ qubits, we can define the group quotient $\cP^n/\cS$. We will usually denote an equivalence class representatives as $\equivclass$, and $[\equivclass]$ is the equivalence class corresponding to $\equivclass$, i.e. $[\equivclass] = \equivclass \cdot \cS$. 
We also introduce a new orthonormal basis $\{\cket{\equivclass,{\syndvec}}\}_{\syndvec \in \Z_2^r}$ for the subspace $\C[\equivclass\cdot \cS] \subset \C[\cP^n]$ \footnote{In \cite{temme2016thermalization}, there is an interesting rationalisation to the introduction of that new basis. Consider a given equivalence class $G \coloneq [\equivclass]$. It is an abelian subgroup isomorphic to $\Z_2^r$, and therefore so is its dual group $\widehat{G}$ -- remember that the dual group also forms a basis for the space $\C[G]$ \cite{babai2002fourier}. Temme observed that the character associated with every $\syndvec \in \Z_2^r$ corresponds to $\chi_{\syndvec}(S\equivclass) = \tr(\Pi_{\syndvec} S)/2^n$, or $\chi_a(Gx\oplus\gamma_0) = e^{i \pi \langle x,a\rangle}$ in the notation of \cite{temme2016thermalization}. Applying standard Fourier theory \cite{babai2002fourier} then gives the following Fourier basis for the space $\C[G]$: ${\{ \Pi_{\syndvec} \equivclass = \sum \tr(\Pi_{\syndvec} S)/2^n S\equivclass \quad :\forall \ \syndvec \in \Z_2^r \}}$. }:
\begin{equation}
    \cket{\equivclass, \syndvec} = \frac{1}{|\cket{\Pi_{\syndvec} \, \equivclass}|}\cket{\Pi_{\syndvec} \, \equivclass}, \quad \cket{\Pi_{\syndvec} \, \equivclass} = \frac{1}{2^n} \sum_{P \in \cP^n} \tr(P \Pi_{\syndvec} \, \equivclass) \cket{P} 
\end{equation}
In order to lighten the notation, we will write $\cket{\equivclass_{\vec{s}}} \coloneq \cket{\equivclass, \syndvec}  $; in a sense the operator $\Pi_{\syndvec} \, \equivclass$ `evaluates' the operator $\equivclass$ in the subspace corresponding to $\Pi_{\syndvec}$. With this new basis we have a way to encode the relationship between $\Pi_{\syndvec}\,\equivclass$ and $\alpha\Pi_{\syndvec}\,\equivclass \alpha = [\alpha, \equivclass] \Pi_{\syndvec \oplus \synd{\alpha}}  \equivclass $, which we encode in the state $\cket{\alpha \circ {\equivclass}_{\vec{s}} }$ defined as follows:
\begin{equation}
    \cket{\alpha \circ {\equivclass}_{\vec{s}} } \coloneq \frac{1}{\sqrt{2}} \Bigl[  \cket{{\equivclass}_{\vec{s}} } -  \theta_{\alpha, \equivclass} \cket{{\equivclass}_{\syndvec \oplus \synd{\alpha}}} \Bigr]
\end{equation}
Where we write $\theta_{\alpha,O} \coloneq [\alpha, O]$. These vectors satisfy a very useful `telescopic' property. 
\begin{proposition}[See Proposition 9 of \cite{temme2016thermalization}]
\label{prop:telescopic}
Let $P \in \cP^n$, and $\gamma = \left(U_{0}, U_{1}, U_{2}, \dots, U_{T} \right )\in \Gamma_P$, then:
    \[\cket{P \circ {\equivclass}_{\vec{s}} } = \sum_{t=0}^{|\gamma|-1} \theta_{U_t, \tilde{O}} \cket{\alpha_{t+1} \circ \tilde{O}_{\syndvec \oplus \synd{U_tQ}}}\]
    Where we define $\alpha_{t+1} = U_t^\dagger U_{t+1}$.
\end{proposition}



With these elements in place we can start discussing the factorisation of $\hat{\cE}$ and $\hat{\cV}$.
\begin{lemma}
\label{lem:matrice-forms}
   Define:
   \begin{equation}
    \hat{\cE}'_{\tilde{O}}(\syndvec) = \sum_{Q \in \cP^n} \sum_\alpha \frac{1}{4}h^\alpha(\omega^\alpha (\syndvec \oplus \sigma_Q))p(\syndvec \oplus \sigma_Q) \cdyad{\alpha \circ \tilde{O}_{\syndvec \oplus \synd{Q}} }
\end{equation}
\begin{equation}
    \hat{\mathcal{V}}_{\tilde{O}}(\syndvec) = \sum_{Q} \sum_{P} {\frac{1}{2^n} \syndromepb{\syndvec \oplus \synd{Q}} \syndromepb{\syndvec \oplus \synd{PQ}
     }} \cdyad{P \circ \tilde{O}_{\syndvec \oplus \synd{Q}}} 
\end{equation}

    Then $\tau \leq \max_{\tilde{O} \in \cP^n/\cS} \max_{\syndvec \in \Z_2^r} \tau_{\tilde{O}}(\syndvec)$, where $\tau_{\tilde{O}}(\syndvec)$ is the support number of the pair $(\hat{\cE}'_{\tilde{O}}(\syndvec), \hat{\mathcal{V}}_{\tilde{O}}(\syndvec))$.
\end{lemma}
\begin{proof}
    The proof is from \cite{temme2016thermalization}, but we sketch it here.
    We define a new matrix $\hat{\cE}'$:
    \[
    \hat{\cE}' \coloneq \frac{1}{4^n} \bigoplus_{\tilde{O} \in \cP^n/\cS} \left( \sum_{\syndvec} \hat{\cE}'_{\tilde{O}}(\syndvec) \right)
    \]
    Which satisfies $\hat{\cE} \geq \hat{\cE}'$. Further one can verify that: 
    \[
    \hat{\cV} = \frac{1}{4^n} \bigoplus_{\tilde{O} \in \cP^n/\cS} \left( \sum_{\syndvec} \hat{\cV}_{\tilde{O}}(\syndvec) \right)
    \]
    Then we have:
    \[
    0 \leq \frac{1}{4^n} \bigoplus_{\tilde{O}} \left( \sum_{\syndvec}  \left (\tau_{\tilde{O}}(\syndvec)\hat{\cE}'_{\tilde{O}}(\syndvec)  - \hat{\cV}_{\tilde{O}}(\syndvec) \right ) \right) \leq \left(\max_{\tilde{O}} \max_{\syndvec} \tau_{\tilde{O}}(\syndvec) \right) \hat{\cE}' - \hat{\cV} 
    \]
    Reading off the definition of support number -- see Definition \ref{def:support-number} -- we get the desired result.
\end{proof}

With Lemma \ref{lem:matrice-forms} giving us matrices that will allow us to bound $\tau$, our next step is to find a $A,W,B$ decomposition for $(\hat{\cE}'_{\tilde{O}}(\syndvec), \hat{\mathcal{V}}_{\tilde{O}}(\syndvec))$ so that we are able to apply Lemma \ref{lem:factorisation-bound}.
\begin{lemma}
\label{lem:factorisation}
    For $\hat{\cE}'_{\tilde{O}}(\syndvec)$ and $\hat{\mathcal{V}}_{\tilde{O}}(\syndvec)$ as above, we can find matrices $A_{\tilde{O}}(\syndvec),B_{\tilde{O}}(\syndvec), W_{\tilde{O}}(\syndvec)$ that allow us to apply to Lemma \ref{lem:factorisation-bound} which are given by:
    \[
    A_{\tilde{O}}(\syndvec) = \sum_Q \sum_{\alpha} \sqrt{\frac{1}{4} h(\omega^\alpha(\syndvec \oplus \synd{Q})) \syndromepb{\syndvec \oplus \synd{Q}} } \cdyadm{\alpha \circ \tilde{O}_{\syndvec \oplus \synd{Q}} }{Q, \alpha Q}
    \]    
    \[
  B_{\tilde{O}}(\syndvec) = 
     \sum_{Q} \sum_{P} \sqrt{\frac{1}{2^n} \syndromepb{\syndvec \oplus \synd{Q}} \syndromepb{\syndvec \oplus \synd{PQ}
     }} \cdyadm{P \circ \tilde{O}_{\syndvec \oplus \synd{Q}}}{Q, PQ} \otimes \frac{1}{\sqrt{\abs{\Gamma_{P}}}} \sum_{\ppath \in \Gamma_{P} } \cbra{\ppath}
  \]
    \[
        W_{\tilde{O}}(\syndvec) = \sum_{Q} \sum_{P}  \frac{1}{\sqrt{\abs{\Gamma_{P}}}}  \sum_{\ppath \in \ppathset_{P}} W_{\tilde{O}}^{Q \xrightarrow{\gamma} P } (\syndvec) \otimes \cbra{\ppath}
    \]
    where for given $\gamma = \left(U_{0}, U_{1}, U_{2}, \dots, U_{T} \right )$ we write $\alpha_{t+1} = U_t^\dagger U_{t+1}$ and:
   \[
     W_{\tilde{O}}^{Q \xrightarrow{\gamma} P } (\syndvec) = \sum_{t=0}^{|\gamma|-1} \sqrt{\frac{4 \syndromepb{\syndvec \oplus \synd{Q}} \syndromepb{\syndvec \oplus \synd{U_TQ}}}{2^n h(\omega^{\alpha_{t+1}}(\synd{U_t Q})) \syndromepb{\syndvec \oplus \synd{U_t Q}}}} \theta_{U_t, \tilde{O}} \rvert\, U_t Q, U_{t+1} Q ) (Q, PQ\vert 
  \]
  and $\{\cket{\gamma}\}_\gamma$ is an orthonormal basis we use to index the different Pauli paths.
\end{lemma}
\begin{proof}
    We first verify the quantities $A_{\tilde{O}}(\syndvec) A_{\tilde{O}}(\syndvec)^\dagger$ and $B_{\tilde{O}}(\syndvec) B_{\tilde{O}}(\syndvec)^\dagger $ yield the desired results:
    \begin{align*}
        A_{\tilde{O}}(\syndvec) A_{\tilde{O}}(\syndvec)^\dagger = \sum_Q \sum_{\alpha} {\frac{1}{4} h(\omega^\alpha(\syndvec \oplus \synd{Q})) \syndromepb{\syndvec \oplus \synd{Q}} } \cdyad{\alpha \circ \tilde{O}_{\syndvec \oplus \synd{Q}} } = \hat{\mathcal{E}}_{\tilde{O}}'(\syndvec)
    \end{align*}
    \begin{align*}
        B_{\tilde{O}}(\syndvec) B_{\tilde{O}}(\syndvec)^\dagger = & \sum_{Q} \sum_{P} {\frac{1}{2^n} \syndromepb{\syndvec \oplus \synd{Q}} \syndromepb{\syndvec \oplus \synd{PQ}
     }} \cdyad{P \circ \tilde{O}_{\syndvec \oplus \synd{Q}}}\otimes \frac{1}{\abs{\Gamma_{P}}} \sum_{\ppath \in \Gamma_{P} } \cbraket{\ppath}{\ppath} \\
     = & \sum_{Q} \sum_{P} {\frac{1}{2^n} \syndromepb{\syndvec \oplus \synd{Q}} \syndromepb{\syndvec \oplus \synd{PQ}
     }} \cdyad{P \circ \tilde{O}_{\syndvec \oplus \synd{Q}}} = \hat{\mathcal{V}}_{\tilde{O}}(\syndvec)
    \end{align*}
    We then address the $AW= B$ condition. Exploiting the telescopic property (see Proposition \ref{prop:telescopic}), we first compute $A_{\tilde{O}}(\syndvec)  W_{\tilde{O}}^{Q \xrightarrow{\gamma} P } (\syndvec) $ :
    \begin{align*}
        A_{\tilde{O}}(\syndvec)  W_{\tilde{O}}^{Q \xrightarrow{\gamma} P } (\syndvec) =& \sum_{t=0}^{|\gamma|-1} \sqrt{\frac{1}{2^n} \syndromepb{\syndvec \oplus \synd{Q}} \syndromepb{\syndvec \oplus \synd{PQ}}} \theta_{U_t, \tilde{O}} \cdyadm{\alpha_{t+1} \circ \tilde{O}_{\syndvec \oplus \synd{U_tQ}}}{Q, PQ} \\
        =&  \sqrt{\frac{1}{2^n} \syndromepb{\syndvec \oplus \synd{Q}} \syndromepb{\syndvec \oplus \synd{PQ}}}  \left( \sum_{t=0}^{|\gamma|-1} \theta_{U_t, \tilde{O}} \cket{\alpha_{t+1} \circ \tilde{O}_{\syndvec \oplus \synd{U_tQ}}} \right ) \cbra{Q, PQ} \\
        = & \sqrt{\frac{1}{2^n} \syndromepb{\syndvec \oplus \synd{Q}} \syndromepb{\syndvec \oplus \synd{PQ}}}   \cdyadm{P \circ \tilde{O}_{\syndvec \oplus \synd{Q}}}{Q, PQ} 
    \end{align*}
    Given the result of the above subcomputation, we can easily verify that $A_{\tilde{O}}(\syndvec) W_{\tilde{O}}(\syndvec) = B_{\tilde{O}}(\syndvec)$:
    \begin{align*}
      A_{\tilde{O}}(\syndvec) W_{\tilde{O}}(\syndvec) = &\sum_{Q} \sum_{P}    \sum_{\ppath \in \ppathset_{P}} A_{\tilde{O}}(\syndvec)  W_{\tilde{O}}^{Q \xrightarrow{\gamma} P } (\syndvec) \otimes \frac{1}{\sqrt{\abs{\Gamma_{P}}}} \cbra{\ppath}  \\
     = & \sum_{Q} \sum_{P} \sqrt{\frac{1}{2^n} \syndromepb{\syndvec \oplus \synd{Q}} \syndromepb{\syndvec \oplus \synd{PQ}}}   \cdyadm{P \circ \tilde{O}_{\syndvec \oplus \synd{Q}}}{Q, PQ} \otimes  \sum_{\ppath \in \ppathset_{P}}  \cbra{\ppath}\\
    = & B_{\tilde{O}}(\syndvec)
    \end{align*}
\end{proof}
Once we have established the factorisation in Lemma \ref{lem:factorisation}, we want to apply Lemma \ref{lem:factorisation-bound}. This requires us to compute the norm of the row vectors of $W_{\tilde{O}}(\syndvec)$.

\begin{proposition}
\label{prop:row-norm}
    Let $W_{\tilde{O}}(\syndvec)$ be as above; then the norm of its row vectors $\ket{w_{(U_eQ, V_eQ)}(\syndvec)} $, as defined in Lemma \ref{lem:factorisation-bound}, can be expressed as:
    \[ \norm{\ket{w_{(U_eQ, V_eQ)}(\syndvec)}}_2^2= \frac{4 }{2^n h\left(\omega^{\alpha_e}\left(\syndvec \oplus \sigma_{U_e Q}\right)\right) \syndromepb{\syndvec \oplus \sigma_{U_e Q}}} \sum_P \frac{\syndromepb{\syndvec \oplus\synd{Q }} \syndromepb{\syndvec \oplus \synd{PQ}}}{\Omega_{P}(e)}\]
    With $\alpha_e = U_e^\dagger V_e$.
\end{proposition}
\begin{proof}
    With the expression of $W_{\tilde{O}}(\syndvec)$ given by Lemma \ref{lem:factorisation}, we obtain the corresponding row vectors by fixing $\cket{U_tQ, U_{t+1}Q}$. We rewrite $U_e \coloneq U_t$, $V_e \coloneq U_{t+1}$ , and $\alpha_e = U_e^\dagger V_e$ to have a path-oblivious expression.
    \begin{align*}
    |w_{(U_eQ, V_e Q)(\syndvec)}\rangle 
    =&\sum_P \sum_{\ppath \in \Gamma_P^{\cap e}} \frac{ \sqrt{ 4 \syndromepb{\syndvec \oplus\synd{Q }} \syndromepb{\syndvec \oplus \synd{PQ}}}}{\sqrt{2^n h\left(\omega^{\alpha_e}\left(\syndvec \oplus \sigma_{U_e Q}\right)\right) \syndromepb{\syndvec \oplus \sigma_{U_e Q}} }}  \theta_{U_e, \tilde{O}}   \cket{Q, PQ}\otimes \frac{1}{\sqrt{\abs{\Gamma_{P}}}} \cket{\ppath}\\
    =&\sum_P \frac{ \sqrt{ 4 \syndromepb{\syndvec \oplus\synd{Q }} \syndromepb{\syndvec \oplus \synd{PQ}}}}{\sqrt{2^n h\left(\omega^{\alpha_e}\left(\syndvec \oplus \sigma_{U_e Q}\right)\right) \syndromepb{\syndvec \oplus \sigma_{U_e Q}} }}  \theta_{U_e, \tilde{O}}   \cket{Q, PQ}\otimes  \sum_{\ppath \in \Gamma_P^{\cap e}}\frac{1}{\sqrt{\abs{\Gamma_{P}}}} \cket{\ppath}\\
    =&\frac{2 \sqrt{\syndromepb{\syndvec \oplus\synd{Q }}} \theta_{U_e, \tilde{O}} }{\sqrt{2^n h\left(\omega^{\alpha_e}\left(\syndvec \oplus \sigma_{U_e Q}\right)\right) \syndromepb{\syndvec \oplus \sigma_{U_e Q}} }} \sum_P    \sqrt{ \syndromepb{\syndvec \oplus \synd{PQ}}} \cket{Q, PQ}\otimes \frac{\sum_{\ppath \in \Gamma_P^{\cap e}} \cket{\ppath}  }{\sqrt{\abs{\Gamma_{P}}}}
\end{align*}
In a second time, computing the $2$-norm does not require much more effort. 
\begin{align*}
    \left\| |w_{(U_eQ, V_eQ)}(\syndvec)\rangle  \right\|_2^2  =&  \frac{4 \syndromepb{\syndvec \oplus\synd{Q }}}{2^n h\left(\omega^{\alpha_e}\left(\syndvec \oplus \sigma_{U_e Q}\right)\right) \syndromepb{\syndvec \oplus \sigma_{U_e Q}}} \left \|  \sum_P     \sqrt{ \syndromepb{\syndvec \oplus \synd{PQ}}} \cket{Q, PQ} \otimes \frac{\sum_{\ppath \in \Gamma_P^{\cap e}} \cket{\ppath}  }{\sqrt{\abs{\Gamma_{P}}}}   \right \|_2^2 \\
    =& \frac{4 \syndromepb{\syndvec \oplus\synd{Q }} }{2^n h\left(\omega^{\alpha_e}\left(\syndvec \oplus \sigma_{U_e Q}\right)\right) \syndromepb{\syndvec \oplus \sigma_{U_e Q}}}  \sum_P    \syndromepb{\syndvec \oplus \synd{PQ}}  \left \|\frac{\sum_{\ppath \in \Gamma_P^{\cap e}} \cket{\ppath}  }{\sqrt{\abs{\Gamma_{P}}}}    \right \|_2^2 \\
    =&  \frac{4 \syndromepb{\syndvec \oplus\synd{Q }}}{2^n h\left(\omega^{\alpha_e}\left(\syndvec \oplus \sigma_{U_e Q}\right)\right) \syndromepb{\syndvec \oplus \sigma_{U_e Q}}}  \sum_P    \syndromepb{\syndvec \oplus \synd{PQ}} \frac{\abs{\Gamma_P^{\cap e}}}{\abs{\Gamma_{P}}}  \\
    =&  \frac{4 \syndromepb{\syndvec \oplus\synd{Q }}}{2^n h\left(\omega^{\alpha_e}\left(\syndvec \oplus \sigma_{U_e Q}\right)\right) \syndromepb{\syndvec \oplus \sigma_{U_e Q}}} \sum_P \frac{\syndromepb{\syndvec \oplus \synd{PQ}}}{\Omega_{P}(e)} 
\end{align*}
\end{proof}

We now can plug Proposition \ref{prop:row-norm} into Lemma \ref{lem:factorisation-bound} to obtain the desired upper bound on the support number. 

\begin{lemma}
\label{lem:support-nb-bound}
    Let $\hat{\cE}$ and $\hat{\cV}$ be the matrices corresponding to the Dirichlet form and variance of the Lindbladian described in Section \ref{sec:model}, then their support number satisfies:
    \[
     \tau \leq \frac{4 n}{c_\circ} e^{\beta \bar{f}}
    \]
\end{lemma}
\begin{proof}
    We start by expressing the bound on $\tau$ obtained from  Proposition \ref{prop:row-norm} into Lemma \ref{lem:factorisation-bound}. Remember that Lemma \ref{lem:matrice-forms} gives us $\tau \leq \max_{\tilde{O} } \max_{\syndvec} \tau_{\tilde{O}}(\syndvec)$. Since $\norm{\ket{w_{(U_eQ, V_eQ)}(\syndvec)}}_2^2$ is independent of $\tilde{O}$, we can abstract away this part of the maximisation:
    \begin{align*}
        \tau &\leq \max_{\tilde{O} } \max_{\syndvec}  \tau_{\tilde{O}}(\syndvec) \\
    & \leq \max_{\syndvec}  \max_{Q, P} \ \max_{\gamma \in \Gamma_P}  \sum_{e \in \gamma}  \norm{\ket{w_{(U_eQ, V_eQ)}(\syndvec)}}_2^2 \\
    &\leq n \max_{\syndvec}  \max_{Q} \ \max_{e} \norm{\ket{w_{(U_eQ, V_eQ)}(\syndvec)}}_2^2 \\
    &= \frac{4n}{2^n} \max_{\syndvec}  \max_{Q} \ \max_{e}  \frac{1}{ h\left(\omega^{\alpha_e}\left(\syndvec \oplus \sigma_{U_e Q}\right)\right) \syndromepb{\syndvec \oplus \sigma_{U_e Q}}}  \sum_{P'}  \frac{\syndromepb{\syndvec \oplus\synd{Q }} \syndromepb{\syndvec \oplus \synd{PQ'}}}{\Omega_{P'}(e)}\\
    & =\frac{4n}{2^n} \max_{\syndvec}  \max_{e}    \frac{1}{h\left(\omega^{\alpha_e}\left(\syndvec \oplus \sigma_{ U_e}\right)\right) \syndromepb{\syndvec \oplus \sigma_{ U_e}}}   \sum_{P'} \frac{ \syndromepb{\syndvec} \syndromepb{\syndvec \oplus \synd{P'}} }{\Omega_{P'}(e)} 
    \end{align*}
    The third line follows from $\sum_{e \in \gamma} (\cdot)\leq n \cdot \max_e (\cdot)$ because the length of $\gamma$ is at most $n$. The last line follows from relabelling $\syndvec \rightarrow \syndvec \oplus \sigma_Q$. Remember that we have $c_\circ \leq h(\cdot)$ for some $c_\circ$. 
    \begin{align*}
    \tau & \leq \frac{4n}{2^n c_\circ}\max_{\syndvec}  \max_{e} \frac{ 1 }{  \syndromepb{\syndvec \oplus \sigma_{ U_e}}} \sum_{P'} \frac{\syndromepb{\syndvec} \syndromepb{\syndvec \oplus \synd{P'}}}{\Omega_{P'}(e)}     \\
    & \leq \frac{4 n}{2^n c_\circ}\max_{\syndvec}  \max_e \frac{  1}{  \syndromepb{\syndvec \oplus \sigma_{ U_e}}}   \sum_{P'} \frac{ \syndromepb{\syndvec \oplus \sigma_{ U_e}} \syndromepb{\syndvec \oplus \synd{U_e P'}}}{\Omega_{P'}(e)} e^{\beta\bar{\epsilon}_{P'}(U_e)} 
    \end{align*}
    The last line follows from the observation that $\bra{\psi_{\syndvec}} U_e H U_e\ket{\psi_{\syndvec}} + \bra{\psi_{\syndvec}} P'^\dagger U_e^\dagger H U_e P' \ket{\psi_{\syndvec}} \leq \bra{\psi_{\syndvec}} H \ket{\psi_{\syndvec}} + \bra{\psi_{\syndvec}} P'^\dagger H P' \ket{\psi_{\syndvec}}  + \bar{\epsilon}_{P'}(U_e)$ and thus $\syndromepb{\syndvec \oplus \sigma_{ U_e}} \syndromepb{\syndvec \oplus \synd{U_e P'}}e^{\beta\bar{\epsilon}_{P'}(U_e)} \geq \syndromepb{\syndvec} \syndromepb{\syndvec \oplus \synd{P'}}$. This finally yields \footnote{Interestingly, the substitution $P' \rightarrow U_eP'$ corresponds to what is described as an `encoding' in \cite{jerrum2003counting,jerrum1996markov}. In these works it is crucial the encoding is injective over its domain. Here it is resolved by the fact that the encoding is an invertible linear map.}:
    \begin{align*}
        \tau & \leq\frac{4 n}{2^n c_\circ}\max_{\syndvec}  \max_e    \sum_{P'} \syndromepb{\syndvec \oplus \synd{P'}} e^{\beta\bar{\epsilon}_{P'}(U_e)  + \log(\frac{1}{\Omega_{P'}(e) } )}\\
    &  \leq \frac{4 n}{c_\circ} e^{\beta \bar{f}} \max_{\syndvec}  \frac{1}{2^n } \sum_{P'}  \syndromepb{\syndvec \oplus \synd{P'}}\\
    & \leq \frac{4 n}{c_\circ} e^{\beta \bar{f}}
    \end{align*}
    Where the last line follows from $\frac{1}{2^n } \sum_{P'}  \syndromepb{\syndvec \oplus \synd{P'}} = 1$, see Equation (91) of \cite{temme2016thermalization}. 
\end{proof}

\begin{theorem}
    \label{thm:final}
    Let $H$ be a stabilizer Hamiltonian with a Pauli flow of free energy $\bar{f}$ at temperature $\beta$. Then for $t_{\mathrm{mix}}$ as defined by Equation \ref{eq:mixing-time}, we have:
    \[
    t_{\mathrm{mix}} \leq \cO(\beta\norm{H}_\infty n e^{\beta \bar{f}})
    \]
\end{theorem}
\begin{proof}
    Straightforward corollary of Lemma \ref{lem:support-nb-bound}.
\end{proof}

In most applications in QEC, the Hamiltonian $H$ contains at most $O(n)$ terms, each with operator norm $1$. Theorem \ref{thm:final} then reads $t_{\mathrm{mix}} \leq \cO(\beta n^2 e^{\beta \bar{f}})$.

\section{Application to Layer Codes}
\label{sec:layer-codes}
Before considering practical applications, we ought to reconsider the definition of a Pauli flow we used to prove Theorem \ref{thm:final}. It is arguably very rigid, as it requires us to ascribe a path to \emph{every} Pauli. Instead, one might hope to leverage the degeneracy provided by the stabiliser group, and only have to specify a path for every equivalence class. We formally show that to be the case in Lemma \ref{lem:pauli-flow-degen} below.

\begin{lemma}
\label{lem:pauli-flow-degen}
    Let $\Gamma_P = \{\gamma\}_\gamma$ be a Pauli flow for a stabiliser Hamiltonian with free energy $\bar{f}$, with the exception that $\gamma$ terminates at $S_\gamma  P$ for some $S_\gamma \in \cS$. Then there exists $\Gamma_P' = \{ \gamma'\}_\gamma$ that satisfies the definition of Definition \ref{def:pauli-flow}. The free energy of the resulting flow obeys: 
    \[
    \bar{f}' \leq \bar{f} + \omega
    \]
    Where $\omega$ is the maximum weight of the terms of the Hamiltonian $H$.
\end{lemma}

\begin{proof}
     As a notational preliminary, for two paths $\gamma = (\idty, P_1, P_2, \dots, P_\gamma)$ and $\zeta = (\idty, Q_1, Q_2, \dots, Q_\zeta)$, we write $\zeta \circ \gamma = (\idty, P_1, P_2, \dots, P,\allowbreak Q_1 P, Q_2P, \dots, Q_\zeta P_\gamma)$ the concatenation of these two paths. 
     
     For any $S \in \cS$, there exists $\zeta(S) = (\idty, Q_1, Q_2, \dots, S)$ a path from $\idty$ to $S$ such that $\bar{\epsilon}_P(Q_t) \leq \omega$ for any $P \in \cP^n$. Because $S_\gamma^2 = \idty$, if we concatenate $\zeta(S_\gamma) \circ \gamma = (\idty, \dots, S_\gamma P, \dots P)$, we get $\Gamma_P' = \{\zeta(S_\gamma) \circ \gamma\}_\gamma$ that fits the requirement of Definition \ref{def:pauli-flow}. The free energy of the resulting flow can then be upper bounded as:
\begin{align*}
    &\max_P \max_{e : e\in \gamma', \gamma' \in \Gamma_P'} \ \bar{\epsilon}_P(U_e) + \frac{1}{\beta}\log\left(\frac{1}{\Omega_{P}(e) } \right) \\ =& \max_P \max_{\gamma \in \Gamma_P} \left(\max_{e\in \gamma} \bar{\epsilon}_P(U_e) + \frac{1}{\beta}\log\left(\frac{1}{\Omega_{P}(e) } \right), \max_{e\in \zeta(S_\gamma) } \bar{\epsilon}_P(U_e S_\gamma P) + \frac{1}{\beta}\log\left(\frac{1}{\Omega_{P}(e) }  \right)\right)  \\
    \leq & \max_P \max_{\gamma \in \Gamma_P} \left(\max_{e\in \gamma} \bar{\epsilon}_P(U_e) + \frac{1}{\beta}\log\left(\frac{1}{\Omega_{P}(e) } \right), \max_{e\in \zeta(S_\gamma) } \bar{\epsilon}_P(U_e S_\gamma P)\right) \\
    = & \max_P \max_{\gamma \in \Gamma_P} \left(\max_{e\in \gamma} \bar{\epsilon}_P(U_e) + \frac{1}{\beta}\log\left(\frac{1}{\Omega_{P}(e) } \right), \max_{e\in \zeta(S_\gamma) } \bar{\epsilon}_P(U_e)\right) \\
    \leq & \max_P \max_{\gamma \in \Gamma_P} \left(\max_{e\in \gamma} \bar{\epsilon}_P(U_e) + \frac{1}{\beta}\log\left(\frac{1}{\Omega_{P}(e) } \right), \omega \right)  \\
    \leq & \bar{f} + \omega
\end{align*}
\end{proof}


We now move on to addressing the Layer Codes. We chose this family as it saturates the energy barrier achievable for 3D stabiliser Hamiltonians. If possessing a high enough $\bar{\epsilon}$ was sufficient to exhibit self-correction, then surely it would have to be the case for these codes as well. We remind the reader that for a code whose stabiliser group is generated by $\{S_i\}_i \subset \cP^n$, then the associated Hamiltonian is $H = -\sum_i S_i$.

We refer the reader to \cite{williamson2024layer} for an in depth description of the construction of these codes. For our purposes we will solely focus that a layer code is made up of many surface codes, joined at their intersections by specific line and dot defects. Therefore, within one of these surface codes (i.e. a `layer') excitations are free to travel at not energy cost. They do however incur an energy penalty when they cross a line defect. In short, for a step $U$ supported on $l$ layers, we can guarantee $\bar{\epsilon}_P(U) \leq \cO(l)$, for any $P \in \cP^n$.
\begin{figure}[H]
    \centering
    \includegraphics[page=1, scale=.6]{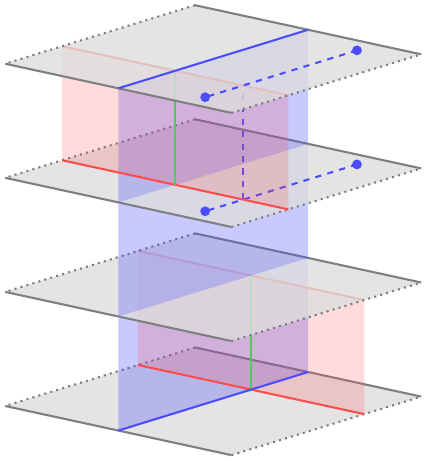}
    \caption{Example of an $X$-type operator on a layer code, and the resulting pattern of excitations. The dashed lines represent the support of the operator, while the dots signal the stabilisers it triggered.}
\end{figure}
Every layer added increases the energy cost, however, as previously mentioned, excitations are free to roam inside a layer, which increases $\Omega(e)$. We will use these observations to build a flow of small free energy for an $X$-type operator $O \in  \{\idty, X\}^{\otimes n}$. 
Let $\mathfrak{L} = \{L_i\}_i$ be the layers the operator is supported on, then for every $i$ we can morph the restriction of $O$ on $L_i$ into a set of $m$ degenerate parallel string operators, see Figure \ref{fig:degen-operators}. 
\begin{figure*}[ht!]
		\centering
		\begin{subfigure}[t]{.31\textwidth}
			\centering
			\includegraphics[page=2, scale=.65]{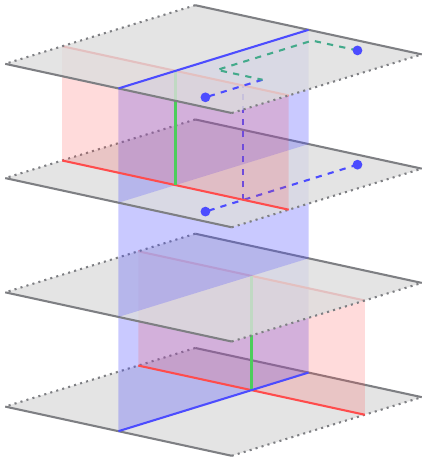}
		\end{subfigure}
		~
		\begin{subfigure}[t]{.31\textwidth}
			\centering
			\includegraphics[page=3, scale=.65]{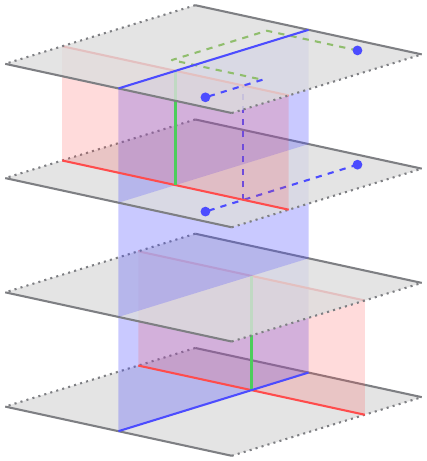}
		\end{subfigure}
        ~
        \begin{subfigure}[t]{.31\textwidth}
			\centering
			\includegraphics[page=4, scale=.65]{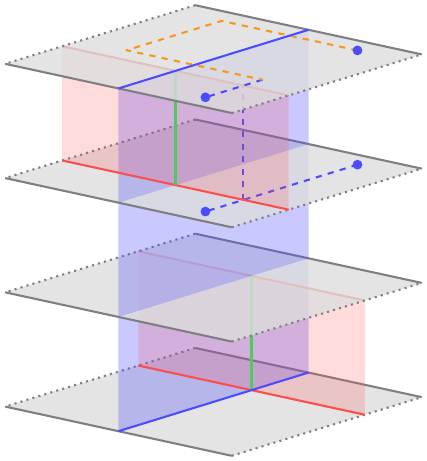}
		\end{subfigure}
		\caption{The stabilisers of the layer codes allow us to move the string operator without paying a significant energy cost.}
        \label{fig:degen-operators}
	\end{figure*}
Note that picking any arbitrary combination $\vec{j} = (j_1,j_2, \dots)$ will result in an operator that differs from $O$ by only a stabiliser. Since there are $m^l$ different $\vec{j}$'s, we obtain a flow $\Gamma_P$ consisting of $m^l$ operators in total. These paths simply consist of building the string operators in the order described by $\vec{j}$. When a path reaches the $l'$-th layer, its edges are shared by only the $m^{l-l'}$ remaining combinations of paths for the layers $l'$ to $l$. We can thus establish that for an edge $e$ on the $l'$-th layer, we have $\Omega_O(e) = \frac{m^l}{m^{l-l'}}$.
\begin{figure}[H]
    \centering
    \includegraphics[page=5, scale=.7]{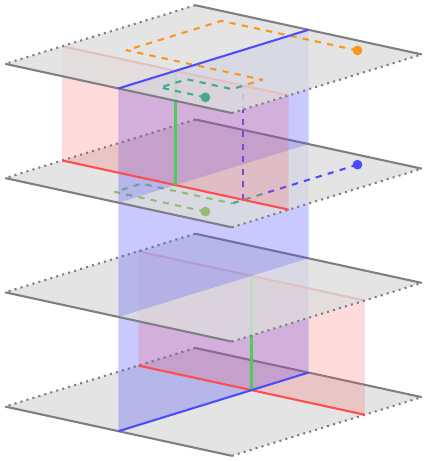}
    \caption{Example of a path from $\Gamma_O$ constructed from using degenerate string operators on each layer.}
\end{figure}
This gives us a flow of free energy $\max_{l':1...l}  \cO(l') - \frac{1}{\beta}l'\log(m) \leq \cO(1)$. This argument can be repeated for $Z$-type operators, and due to the CSS nature of these codes, this is sufficient to upper bound the free energy barrier. Applying Theorem \ref{thm:final}, we obtain that $t_{\mathrm{mix}} \lesssim \beta n^2 e^{\cO(\beta)} $ .

\section{Acknowledgements}
I benefited from generous explanations from Kristian Temme on \cite{temme2016thermalization}; and Christopher Chubb was kind enough to read a preliminary version of this work. 
This work was partially completed while N.\,B. was supported by the Australian Research Council via the Centre of Excellence in Engineered Quantum Systems (EQUS) project number CE170100009, and by the Sydney Quantum Academy.

\bibliographystyle{alpha}
\bibliography{references}

\newcommand{\etalchar}[1]{$^{#1}$}
\begin{thebibliography}{AHHH08}

\bibitem[AFH09]{alicki2009thermalization}
R~Alicki, M~Fannes, and M~Horodecki.
\newblock On thermalization in kitaev’s 2d model.
\newblock {\em Journal of Physics A: Mathematical and Theoretical}, 42(6):065303, January 2009.

\bibitem[AHHH08]{alicki2008thermal}
R.~Alicki, M.~Horodecki, P.~Horodecki, and R.~Horodecki.
\newblock On thermal stability of topological qubit in kitaev's 4d model, 2008.

\bibitem[Alh23]{alhambra2023quantum}
\'Alvaro~M. Alhambra.
\newblock Quantum many-body systems in thermal equilibrium.
\newblock {\em PRX Quantum}, 4:040201, Nov 2023.

\bibitem[Bab02]{babai2002fourier}
L\'aszl\'o Babai.
\newblock The fourier transform and equations over finite abelian groups, June 2002.

\bibitem[BH03]{boman2003support}
Erik~G. Boman and Bruce Hendrickson.
\newblock Support theory for preconditioning.
\newblock {\em SIAM Journal on Matrix Analysis and Applications}, 25(3):694--717, 2003.

\bibitem[BT09]{bravyi2009no}
Sergey Bravyi and Barbara Terhal.
\newblock A no-go theorem for a two-dimensional self-correcting quantum memory based on stabilizer codes.
\newblock {\em New Journal of Physics}, 11(4):043029, 2009.

\bibitem[CGT05]{chen2005obtaining}
Doron Chen, John~R. Gilbert, and Sivan Toledo.
\newblock Obtaining bounds on the two norm of a matrix from the splitting lemma.
\newblock {\em ETNA. Electronic Transactions on Numerical Analysis [electronic only]}, 21:28--46, 2005.

\bibitem[Dav74]{Davies1974}
E.~B. Davies.
\newblock Markovian master equations.
\newblock {\em Communications in Mathematical Physics}, 39(2):91–110, June 1974.

\bibitem[Dav76]{Davies1976}
E.~B. Davies.
\newblock Markovian master equations. ii.
\newblock {\em Mathematische Annalen}, 219(2):147–158, June 1976.

\bibitem[Gur16]{guruswami2016rapidly}
Venkatesan Guruswami.
\newblock Rapidly mixing markov chains: A comparison of techniques (a survey), 2016.

\bibitem[ILP{\etalchar{+}}23]{innocenti2023shadow}
L.~Innocenti, S.~Lorenzo, I.~Palmisano, F.~Albarelli, A.~Ferraro, M.~Paternostro, and G.~M. Palma.
\newblock Shadow tomography on general measurement frames.
\newblock {\em PRX Quantum}, 4(4), November 2023.

\bibitem[Jer03]{jerrum2003counting}
Mark Jerrum.
\newblock {\em Counting, Sampling and Integrating: Algorithm and Complexity}.
\newblock Birkh\"{a}user Basel, 2003.

\bibitem[JS96]{jerrum1996markov}
Mark Jerrum and Alistair Sinclair.
\newblock {\em The Markov chain Monte Carlo method: an approach to approximate counting and integration}, page 482–520.
\newblock PWS Publishing Co., USA, 1996.

\bibitem[KGL23]{kasatkin2023whichdifferential}
Victor Kasatkin, Larry Gu, and Daniel~A. Lidar.
\newblock Which differential equations correspond to the lindblad equation?
\newblock {\em Phys. Rev. Res.}, 5:043163, Nov 2023.

\bibitem[KT13]{kastoryano2013quantum}
Michael~J. Kastoryano and Kristan Temme.
\newblock Quantum logarithmic sobolev inequalities and rapid mixing.
\newblock {\em Journal of Mathematical Physics}, 54(5):052202, 05 2013.

\bibitem[LWH24]{lin2024geometricallylocalquantumclassical}
Ting-Chun Lin, Adam Wills, and Min-Hsiu Hsieh.
\newblock Geometrically local quantum and classical codes from subdivision, 2024.

\bibitem[Man20]{manzano2020short}
Daniel Manzano.
\newblock A short introduction to the lindblad master equation.
\newblock {\em AIP Advances}, 10(2), February 2020.

\bibitem[Mic12]{michnicki2012welded}
Kamil Michnicki.
\newblock 3-d quantum stabilizer codes with a power law energy barrier, 2012.

\bibitem[ML20]{mozgunov2020completelypositive}
Evgeny Mozgunov and Daniel Lidar.
\newblock Completely positive master equation for arbitrary driving and small level spacing.
\newblock {\em {Quantum}}, 4:227, February 2020.

\bibitem[MT06]{montenegro2006mathematical}
Ravi Montenegro and Prasad Tetali.
\newblock Mathematical aspects of mixing times in markov chains.
\newblock {\em Found. Trends Theor. Comput. Sci.}, 1(3):237–354, July 2006.

\bibitem[Por23]{portnoy2023localquantumcodessubdivided}
Elia Portnoy.
\newblock Local quantum codes from subdivided manifolds, 2023.

\bibitem[PRBK24]{placke2024topologicalquantumspinglass}
Benedikt Placke, Tibor Rakovszky, Nikolas~P. Breuckmann, and Vedika Khemani.
\newblock Topological quantum spin glass order and its realization in qldpc codes, 2024.

\bibitem[RFZ10]{roga2010davies}
Wojciech Roga, Mark Fannes, and Karol Zyczkowski.
\newblock Davies maps for qubits and qutrits.
\newblock {\em Reports on Mathematical Physics}, 66(3):311–329, December 2010.

\bibitem[RH12]{rivas2012quantum}
Angel Rivas and Susana~F. Huelga.
\newblock {\em Quantum Markov Process: Mathematical Structure}, pages 33--48.
\newblock Springer Berlin Heidelberg, Berlin, Heidelberg, 2012.

\bibitem[Rob03]{roberto2003path}
Cyril Roberto.
\newblock A path method for the logarithmic sobolev constant.
\newblock {\em Comb. Probab. Comput.}, 12(4):431–455, July 2003.

\bibitem[RPBK24]{rakovszky2024bottlenecksquantumchannelsfinite}
Tibor Rakovszky, Benedikt Placke, Nikolas~P. Breuckmann, and Vedika Khemani.
\newblock Bottlenecks in quantum channels and finite temperature phases of matter, 2024.

\bibitem[Sin92]{sinclair1992improved}
Alistair Sinclair.
\newblock Improved bounds for mixing rates of markov chains and multicommodity flow.
\newblock In Imre Simon, editor, {\em LATIN '92}, pages 474--487, Berlin, Heidelberg, 1992. Springer Berlin Heidelberg.

\bibitem[SY17]{siva2017topological}
Karthik Siva and Beni Yoshida.
\newblock Topological order and memory time in marginally-self-correcting quantum memory.
\newblock {\em Physical Review A}, 95(3), March 2017.

\bibitem[Tem13]{temme2013lowerbounds}
Kristan Temme.
\newblock Lower bounds to the spectral gap of davies generators.
\newblock {\em Journal of Mathematical Physics}, 54(12):122110, 12 2013.

\bibitem[Tem16]{temme2016thermalization}
Kristan Temme.
\newblock Thermalization time bounds for pauli stabilizer hamiltonians, 2016.

\bibitem[WB24]{williamson2024layer}
Dominic~J. Williamson and Nouédyn Baspin.
\newblock Layer codes.
\newblock {\em Nature Communications}, 15(1), November 2024.

\end{thebibliography}

\end{document}